\documentclass[letterpaper, 10 pt, conference]{ieeeconf} 
\IEEEoverridecommandlockouts                              
\usepackage[english]{babel}

\usepackage{amsmath, amssymb}
\usepackage{enumerate}
\usepackage{graphicx}  
\usepackage[dvipsnames]{xcolor}

\usepackage{subcaption}  

\usepackage{algorithm}
\usepackage{algpseudocode}
\algrenewcommand\algorithmicrequire{\textbf{Input:}}
\algrenewcommand\algorithmicensure{\textbf{Output:}}

\usepackage[overload]{empheq}
\usepackage[noadjust]{cite}

\usepackage[hyperindex=true, breaklinks=true,colorlinks=true,bookmarks=false, pdftitle={Towards Closed-loop Stability of Nonlinear Receding Horizon Games},pdfauthor={Authors}, pagebackref=false, plainpages=false, pdfpagelabels, linkcolor=black, citecolor=black, filecolor=black, urlcolor=black]{hyperref}

\definecolor{lightgray}{gray}{0.9}
\definecolor{ggreen}{rgb}{0,0.5,0}
\definecolor{rred}{rgb}{0.5,0,0}

\newcommand{\bb}[1]{\mathbb{#1}}

\newcommand{\mc}[1]{\mathcal{#1}}

\newcommand{\R}{\mathbb{R}}
\newcommand{\x}{\mathbf{x}}

\newcommand{\eDef}{$\hfill\blacksquare$}

\DeclareMathOperator{\subjectto}{\ {s.}{t.}\ }
\DeclareMathOperator{\dist}{dist}
\DeclareMathOperator{\inte}{int}

\newtheorem{ass}{Assumption}
\newtheorem{dfn}{Definition}
\newtheorem{lmm}{Lemma}
\newtheorem{prp}{Proposition}
\newtheorem{thm}{Theorem}
\newtheorem{rmk}{Remark}

\newenvironment{proof}{\paragraph*{Proof}}{\hfill$\square$}

\newcommand{\sh}[1]{\textcolor{NavyBlue}{[#1]\raise 0.5ex \hbox{\footnotesize{SH}}}}
\newcommand{\CHECK}[1]{\textcolor{Red}{[#1]\raise 0.5ex \hbox{\footnotesize{CHECK}}}}
\newcommand{\tf}[1]{\textcolor{Orange}{[#1]\raise 0.5ex \hbox{\footnotesize{TF}}}}

\graphicspath{{figs/}}

\title{\LARGE \bf Towards Closed-loop Stability of Nonlinear Receding Horizon Games} 
\author{Sophie Hall, Florian Dörfler, Timm Faulwasser
\thanks{The authors are with the  Automatic Control Lab, ETH Z\"urich, Physikstrasse 3, 8092 Z\"urich, {Emails: \texttt{\{shall, dorfler\}@ethz.ch}}
, Switzerland and with the Institute of Control Systems, Hamburg University of Technology, Harburger Schlo{\ss}stra{\ss}e 22a, 21079 Hamburg, Germany, {Email: \texttt{timm.faulwasser@ieee.org}}. This work was supported by the Swiss National Science Foundation under the NCCR Automation (grant 51NF40 225155).}
}

\begin{document}

\maketitle
\thispagestyle{empty}
\pagestyle{empty}

\begin{abstract}
We analyze Receding Horizon Games without any MPC-like terminal ingredients. We show that recursive feasibility can be inferred from the turnpike phenomenon under mild assumptions. Moreover, we prove sufficient conditions for practical asymptotic convergence of the closed-loop trajectories, and we discuss how the gap towards practical asymptotic stability may be closed. We use numerical examples to show that the closed-loop region of attraction around the steady-state GNE shrinks exponentially with the horizon length, a behavior previously known only for model predictive control. Further, we apply a linear end penalty and demonstrate in numerical simulations that it suppresses the leaving arc and ensures asymptotic convergence to the steady-state GNE.
\end{abstract}

\section{Introduction}

Receding Horizon Games (also known as game-theoretic MPC) is a control framework combining dynamic game theory with receding horizon planning, modeling interactions between self-interested agents coupled through costs, dynamics, and constraints. Such interactions arise naturally in competitive settings, e.g., supply chains where manufacturers repeatedly plan over a prediction horizon, observe competitors' decisions, and re-plan \cite{hall2024receding}. Being a natural model for real-life competitive interactions, RHGs have been employed in 
robotics~\cite{gu2008differential}, autonomous driving~\cite{liniger2020noncooperative, wang2021game, lecleach2022algames}, electric vehicle charging~\cite{mignoni2023distributed}, and smart grids~\cite{paola2018distributed, hall2022receding}.


The underlying solution concept of RHGs is the generalized Nash equilibrium (GNE) \cite{facchinei2009generalized}, an equilibrium at which no agent can unilaterally improve its cost (Nash), and where each agent's feasible set depends on the actions of others (generalized Nash). GNEs can generally be formulated as quasi-variational inequalities. When the variational GNE is sought, which imposes equal shadow prices on shared constraints, the problem reduces to a variational inequality \cite{facchinei2009nash,kulkarni2012variational}. As GNEs originated in the economics and game theory community, they were initially a static solution concept. Only recently, motivated by planning and control applications, have \textit{dynamic} GNEs been developed~\cite{atzeni2013demand, benenati2023optimal, lecleach2022algames}.


A fundamental concept in optimal control to derive system-theoretic properties of dynamic trajectories has been dissipativity of open systems as introduced by Jan C. Willems~\cite{Willems72a} and its link to the turnpike phenomenon. This phenomenon is characterized by the solutions of OCPs clustering near one common steady state across varying initial conditions and horizon lengths~\cite{Mckenzie76,faulwasser2022turnpike}. The link between turnpikes and dissipativity enabled major breakthroughs in the closed-loop stability analysis of nonlinear and economic MPC, see, e.g., \cite{gruene2013economic}. Early works exploiting dissipativity for closed-loop analysis include \cite{diehl2011lyapunov, angeli2012average}; an overview is given in \cite{faulwasser2018economic}. The close relations between turnpikes and dissipativity notions were made explicit in~\cite{gruene2016relation} for discrete-time and in~\cite{epfl:faulwasser15h} for continuous-time systems.

The study of turnpike phenomena in games commenced in the 1980s~\cite{fershtman1986turnpike}, with early results for infinite-horizon open-loop games addressing existence, uniqueness, and convergence to the turnpike~\cite{carlson1995turnpike, carlson1996turnpike, carlson2000infinite}. Recently, turnpike properties in games have received renewed attention, including results for stochastic differential LQ games~\cite{li2025turnpike, cohen2025turnpike} and mean field games under large population assumptions~\cite{cirant2021long, carmona2024leveraging, ersland2025long, fedorov2025studying}. In \cite{hall2025system}, we developed an analysis of turnpike and dissipativity properties in noncooperative games for a general class of nonlinear costs, coupled constraints, and nonlinear dynamics.  

While the above results characterize open-loop dynamic game trajectories, the fundamental question in receding horizon games is whether such trajectories lead to stable closed-loop systems. Closed-loop stability of game-theoretic MPC has been proven under a potential game assumption in~\cite{hall2022receding}. The first results for the non-potential linear-quadratic case have been presented in~\cite{hall2025stability}, applicable to pre-stabilized LTI systems. Following this, \cite{benenati2025linear} proved stability for LQ games under both open-loop and feedback Nash equilibria.


In this paper, we present the first closed-loop stability results for nonlinear RHG with state and input constraints, building on the connection between turnpike and dissipativity theory established in \cite{hall2025system}. We prove recursive feasibility of the RHG feedback law under a cheap reachability and local controllability assumption. Based on this, we show practical asymptotic convergence to the steady-state GNE using a Lyapunov candidate defined as the sum of an agent performance measure and the storage function. Convergence can be strengthened to practical asymptotic stability given an upper bound on the Lyapunov candidate. Using numerical examples, we characterize the convergence neighborhood to the steady-state GNE and show that by applying a terminal penalty, exact asymptotic convergence can be achieved.

\subsection*{Notation} We denote by $\bb{Z}_K= \{0,\dots, K\!-\!1\}$ the sequence of the first $K$ non-negative integers. Given $M$ vectors $u^1, ..., u^M$, we denote  by  $ u = \text{col}(u^v)_{v=1}^M:= [(u^1)^\top, \ldots, (u^M)^\top]^\top$ the stacked vector of vectors $u^v$, where $u^v$ is the decision vector of agent $v$, and of all other agents as $u^{-v} = \text{col}(\{u^{j}\}_{j\in \mc{V}\backslash v})$. Our use of class $\mc{K}$, $\mc{L}$, and $\mc{K}_\infty$ comparison functions follows standard conventions~\cite{Kellett14}. For a finite set $\mc{Q}$, we denote its cardinality by $\#\mc{Q}$. Let $\mc B_\varepsilon(\x) \subset \R^n$ denote a closed ball of radius $\varepsilon$, centered at $\x$, and $\x\in \R^{n_x}$ refers to a point and not a trajectory.

\section{Nonlinear Receding Horizon Games}

We consider a group of self-interested agents $v\in \mc{V}:=\{1,\dots,M\}$ which controls the following shared dynamics
\begin{equation}\label{eq:Dynamics}
    x_{t+1} = f(x_t, u_t^v,u_t^{-v}).
\end{equation}
Each agent minimizes its accumulated stage cost $\ell^v$ over a prediction horizon $N$ with states and inputs coupled to other agents. This constitutes a \emph{finite-horizon Generalized Nash Equilibrium problem} (GNEP) as follows
\begin{subequations}
\label{eq:MPCPerAgent}
\begin{empheq}[left=\forall v\in \mc{V}:  \empheqlbrace]{align}
\label{eq:RunningCost}
\displaystyle \min_{u^v,\, x}  &\;\sum_{k= 0}^{N-1} \ell^v(x_k, u_k^v,u_k^{-v})\\
\textrm{s.t.} \quad &  x_{k+1} =  f(x_k,u_k^v,u_k^{-v})  \hspace{1.25em}  k \in \bb{Z}_{N} \label{eq:Constr1}\\
&g(x_k,u_k^v, u^{-v}_k) \leq 0, \hspace{2.5em} k \in \bb{Z}_{N} \label{eq:Constr2}\\
  & h^v(u_k^v)\leq 0,      \hspace{5.6em} k \in \bb{Z}_{N} \label{eq:Constr3}\\ 
&\; x_0 = \x, \label{eq:Constr4}
\end{empheq}
\end{subequations}
with initial condition $\x$ and nonlinear coupled and local constraints~\eqref{eq:Constr2}-\eqref{eq:Constr3}. Note that  in the following we interchangeably use the notation $f(x_k,u_k^v,u_k^{-v})$ and $f(x_k,u_k)$ for all functions and sets, whereby $f(x_k,u_k)$ is a shorthand referring to the entire population of agents. We introduce the cumulative cost \[J_N^v(x,u^v, u^{-v}) :=\sum_{k= 0}^{N-1} \ell^v(x_k, u_k^v, u^{-v}_k),\]
the global feasible set
\begin{subequations}
\begin{align}
\mc{Z}_N = \{(x, u) \in \R^{(N+1)n_x+N n_u}~|~  \eqref{eq:Constr1} - \eqref{eq:Constr3}\},  
\end{align}
the per-agent and global feasible sets as a function of the initial condition $\x$
\begin{align}
&\mc{Z}_N^v( \x, u^{-v}) = \{(x,u^v)  ~|~ \eqref{eq:Constr1} - \eqref{eq:Constr4}\}, \\[5pt]  
&\mc{Z}_N(\x) = \{(x, u) \in \R^{(N+1)n_x+N n_u}~|~  \eqref{eq:Constr1} - \eqref{eq:Constr4}\} \label{eq:FeasibleSet},
\end{align}
\end{subequations}
and similarly $\mc{Z}_{\infty}^v$ and $\mc{Z}_{\infty}(\x)$ for the infinite-horizon setting. Clearly, $\mc{Z}_N(\x) \subset \mc{Z}_N$. The projection of $\mc Z_N$ onto the state-input space $\R^{n_x + n_u}$ is written as $\bb Z$, while the projection onto the state space $\R^{n_x}$ is denoted as $\bb Z_x$. 

Decisions that jointly solve~\eqref{eq:MPCPerAgent} are called \textit{generalized Nash equilibria} (GNE)~\cite[\S 2]{facchinei2009nash}. Intuitively, at a GNE no agent $v \in \mc{V}$ can reduce its cost by unilaterally changing its own decision as defined next.
\begin{dfn}[Generalized Nash equilibrium] \label{dfn:GNE} ~\\A joint decision $(x^*, u^*)\in \mc{Z}_N(\x)$ is a GNE  of~\eqref{eq:MPCPerAgent} if 
\begin{align*}
\forall v\in \mc{V}: \, J_N^v(x^{*}, u^{v*}, u^{-v*} ) \leq J_N^v(x, u^{v}, u^{-v*}) 
\end{align*} holds for all $(x,u^v) \in \mc{Z}_N( \x, u^{-v*})$. The corresponding solution set for fixed $N\in \bb{N}$ and $\x\in \bb{X}_0$ is denoted as
\[
(x^*,u^*) \in \mc{S}^{\text{\tiny GNE}}_N(\x) \subset \R^{(N+1)n_x + N n_u}.
\]
We refer to any $(x^*,u^*) \in \mc{S}^{\text{\tiny GNE}}_N(\x)$ as a \emph{game pair}.
\eDef \vspace{2mm}
\end{dfn}

We recursively solve the finite-horizon GNEP~\eqref{eq:MPCPerAgent}, and apply its solution in a receding-horizon fashion. Specifically, at each sampling time $t$, the agents compute the GNE of the game~\eqref{eq:MPCPerAgent} with measured state $\x_t$ and then apply the first element $u^*_{0|t}$ of the optimal control trajectory. This defines the implicit feedback policy (referred to as the \textit{receding-horizon game} (RHG) feedback law) 
\begin{align}\label{eq:FeedbackLaw}
u^*_{0|t} = \mathrm{col}(u_{0|t}^{*,v})_{v \in \mc{V}}  
=:  \mu^*(\x_t) = \kappa(\mc{S}^{\text{\tiny GNE}}_N(\x_t)),
\end{align}
where $\kappa$ is a deterministic discrete selection mechanism, returning a unique GNE out of the possibly infinite set $\mc{S}^{\text{\tiny GNE}}_N(\x_t)$ which  extracts the first element of the control sequence of each agent $v\in \mc{V}$, namely, $\{u_{{k=0|t}}^{*,v}\}_{v \in \mathcal{V}}$. This selection mechanism is necessary as we do not impose conditions which ensure uniqueness of the GNE in~\eqref{eq:MPCPerAgent}. Various selection mechanisms exist in the literature~\cite{benenati2023optimal, hall2025limits, hall2025solving}.  Whenever necessary, we clearly differentiate predicted trajectories as
%
\begin{align*}
    x^*_{k|t}(\x_t) \quad\text{and}\ \quad u^*_{k|t}(\x_t),
\end{align*}
where the subscript $\cdot_{k|t}$ refers to the instant $t$ at which the trajectory is computed and $k \in \{0, \dots, N\}$ is the prediction step. When we refer to the entire predicted sequence we denote it as $\cdot_{\cdot|t}$ and when clear from context, we suppress the dependence on the initial condition, i.e., $x^*_{k|t}$.

Subsequently, we study the asymptotics of the  closed-loop RHG dynamics, i.e., the dynamics of applying the feedback \eqref{eq:FeedbackLaw} to \eqref{eq:Dynamics}. In particular, we are interested in characterizing the closed-loop stability with respect to a steady-state GNE which is a twofold equilibrium: (i) a steady-state of~\eqref{eq:Dynamics}, i.e., $\bar x = f(\bar x, \bar u^v, \bar u^{-v})$; and (ii) a strategic (decision) equilibrium of the one-step GNEP in~\eqref{eq:MPCPerAgent} as defined next.

\begin{dfn}[Steady-state GNE] The pair $(x_s,u_s)$ is called a \emph{steady-state GNE} if it solves 
\begin{align}\label{eq:SteadyStateGNEP}
v\in \mc{V}: \left\{
\begin{array}{r l}
\displaystyle \min_{\bar{u}^v, \bar{x}} & \; \ell^v(\bar{x}, \bar{u}^v, \bar{u}^{-v}) \\ 
 \subjectto  &  f(\bar{x},\bar{u}^v,\bar{u}^{-v})- \bar{x}=0\\ 
            &  g(\bar{x},\bar{u}^v, \bar{u}^{-v}) \leq 0,\\
            &h^v(\bar u^v)\leq 0,\\
\end{array} 
\right.
\end{align}
with the corresponding solution set $\mc{S}^{\text{\tiny GNE}}_s \subset \R^{n_x + n_u}.$\eDef
\end{dfn}

\section{Turnpike and Dissipativity of RHGs}

The turnpike phenomenon refers to a similarity property of parametric optimal control problems and has been established as a crucial element in the closed-loop analysis of  model predictive control (MPC), see, e.g.~\cite{gruene2013economic}. In games, the turnpike phenomenon was first observed in economics in the 1980s~\cite{fershtman1986turnpike} and analyzed for infinite-horizon open-loop games in continuous time~\cite{carlson1995turnpike} and discrete time~\cite{carlson1996turnpike}. In the context of RHGs, turnpikes have been observed in competitive dynamic supply chains~\cite{hall2024receding}. In the following we recall some fundamental results connecting dissipativity of GNEPs to the turnpike property which were first presented in~\cite{hall2025system} and are the foundation of our closed-loop analysis. 

We introduce the following performance measure for the group of agents
\begin{equation}\label{eq:overallJ}
  J_N(x,u) :=\sum_{k= 0}^{N-1} \ell(x_k, u_k)
  = \sum_{k= 0}^{N-1} \sum_{v\in \mathcal V} \ell^v(x_k,u_k^v, u_k^{-v})
\end{equation}
as well as the 
\textit{game value function} $V_N^*:\bb X_0 \to \R$ which measures the performance of the agent population at $\x$ for the GNE trajectory $(x^*,u^*)\in \mc{S}^{\text{\tiny GNE}}_N(\x)$
\begin{equation}\label{eq:Vfun}
  \hspace{-1.5mm}  V^*_N(\x) := \sum_{k=0}^{N-1}\ell(x_k^*, u^*_k)
    =\sum_{k=0}^{N-1}  \sum_{v\in \mc V}\ell^v(x^*_k, u_k^{v*}, u_k^{-v*}).
\end{equation}
We refer to~\cite{hall2025system} for recent results analyzing the properties of the game value function. 
To make use of the connection between the turnpike property and dissipativity derived in~\cite{hall2025system}~we require the following assumptions.   
\begin{ass}[Cheap reachability on $\bb X_0$] \label{ass:cheap}
Let the set $\bb X_0$ be the largest subset of $\bb Z_x$ such that the following holds: For 
any initial condition $\x\in \bb X_0$ there exists an infinite-horizon feasible pair $(x,u) \in \mc{Z}_\infty(\x)$ 
    such that, for some $\delta \in\R$ and $\forall N\in \bb N$ it holds that
    \[J_N(x,u)\leq \; N \ell(x_s, u_s) + \delta.
    \]
    We assume that  $\bb X_0 \not= \emptyset$.
\eDef \vspace{2mm}
\end{ass}
Note that 
$(x,u)$ in Assumption~\ref{ass:cheap} does not need to be a game pair. 

Further, we denote as $\bb X_N(\bb{X}_0)\subseteq \R^{n_x}$\footnote{Specifically, if $\tilde x \in \bb X_N(\bb{X}_0)$ then $\exists \;\x \in \bb X_0$ such that $(x^*(\x),u^*(\x)) \in \mc{S}^{\text{\tiny GNE}}_N(\x) $  and $x^*(\x)$ passes through $\tilde x$ at least once.} the point-wise in time projection of $\mc{S}^{\text{\tiny GNE}}_N(\bb X_0)$ onto the states and similarly $\bb X_{\infty}(\bb{X}_0)$ for the infinite-horizon setting.

The optimal control problem for the entire agent population using the cost function \eqref{eq:overallJ} reads
  \begin{align}\label{eq:OCP}\;
        V^\diamond_N(\x) := \min_{u,x}\; J_N(x,u) 
        \subjectto (x,u)\in \mc{Z}_N(\x),
    \end{align}
where $ V^\diamond_N(\x)$ is the usual OCP value function. Henceforth we use the superscript $\cdot^\diamond$ to highlight optimal quantities obtained from solving the OCP \eqref{eq:OCP}, while the superscript $\cdot^{*}$ refers to solutions of the GNEP~\eqref{eq:MPCPerAgent}.
Indeed, the GNEP and OCP solutions generally do not coincide and clearly \[V_N^*(\x) \geq V_N^\diamond(\x).\] The maximal loss induced by the self-interested behavior of agents in the game-theoretic setting is also called the \textit{price of anarchy}. We impose an assumption of it being bounded.
\begin{ass}[Bounded price of anarchy] \label{ass:PoA}
For all $\x\in\bb X_0$, the price of anarchy satisfies
\[
\text{PoA}(\x) :=   \frac{ \sup_{(x^*, u^*) \in \mc{S}^{\text{\tiny GNE}}_N(\x)} J_N(x^*, u^*)}{V^\diamond_N(\x)}   \leq P < \infty
    \]
    and $0<\nu  \leq V^\diamond_N(\x) \leq V<\infty$ holds for any $N\in \bb{N}$.
\eDef\end{ass}
Note that this also implies $V^*_N(\x)\leq P\,  V^\diamond_N(\x)$ as $ V^*_N(\x) \leq \sup_{(x^*, u^*) \in \mc{S}^{\text{\tiny GNE}}_N(\x)}  J_N(x^*, u^*)$. 

\begin{lmm}[{\cite{hall2025system}}]\label{lem:linearBound}~\\
Suppose that $0 < \nu \leq V^\diamond_N(\x) \leq V< \infty$ is satisfied. 
\begin{itemize}
    \item[(i)]If Assumption~\ref{ass:PoA} holds with $P\in \R$, then 
 \begin{equation}\label{eq:linearBound}
     \sup_{(x^*, u^*) \in \mc{S}^{\text{\tiny GNE}}_N(\x)} J_N(x^*, u^*) - V^\diamond_N(\x) \leq VP =:\bar P.
 \end{equation} 
 \item[(ii)] If \eqref{eq:linearBound} holds with bound $VP = \bar P$, then Assumption~\ref{ass:PoA} holds with $\frac{V}{\nu}P +1$.\eDef
\end{itemize}
\end{lmm}
We note that, if $0 < \nu \leq V^\diamond_N(\x)$ does not hold for the stage cost $\ell$, there is a simple remedy, namely, adding a positive bounded offset ($\ell +c$) will ensure the lower bound on  $V^\diamond_N(\x)\geq \nu$ for any finite horizon $N$. Moreover, for any fixed horizon $N\in \bb{N}$, $V_N^\diamond$ will be finite. Hence assuming a finite upper bound $V\geq V^\diamond_N(\x)$ is not restrictive.

Using the performance measure of the agent population in~\eqref{eq:overallJ} we define a strict dissipativity notion for game-theoretic settings.
 \begin{dfn}[Strict dissipativity of GNEPs \cite{hall2025system}]\label{dfn:StrictDiss}
Given a steady-state GNE $(x_s,u_s)\in \mc{S}^{\text{\tiny GNE}}_s$, the GNEP~\eqref{eq:MPCPerAgent} is called strictly dissipative with supply rate 
\[s(x^*_k,u^*_k) := \ell(x^*_k,u^*_k) -\ell(x_s,u_s)\]
if
there exists a storage function $\Lambda: \bb X_N(\bb{X}_0)\to \R$ bounded from below, such that $\forall N \in \bb{N}, \forall \x \in \bb{X}_0$
 \begin{multline}\label{eq:GsDI}
\Lambda(f(x^*_k,u^*_k)) -\Lambda(x^*_k) \leq \\ -  \alpha_{\ell}\left(\left\|
\begin{matrix}x_k^*- x_s\\ u_k^* - u_s \end{matrix}
\right\|\right) +s(x^*_k,u^*_k) \tag{sDI}
 \end{multline}
holds for some $\alpha_{\ell}\in \mc{K}$ and each point $(x^*_k,u^*_k)$ along game pairs $(x^*,u^*)\in \mc{S}^{\text{\tiny GNE}}_N(\x)$. \eDef
 \end{dfn}

The formal definition of the turnpike property of the GNEP in~\eqref{eq:MPCPerAgent} with respect to the steady-state GNE~\eqref{eq:SteadyStateGNEP} is given next. 
\begin{dfn}[Measure turnpike in GNEPs \cite{hall2025system}]\label{dfn:GameTurnpike}~\\
%
The GNEP~\eqref{eq:MPCPerAgent} exhibits the (measure) turnpike property at $(x_s,u_s)$ if for each $\varepsilon>0$ there exists $C>0$ such that $\forall N\in \bb{N}, \forall \x\in  \bb{X}_0$, and for all game pairs $(x^*,u^*)\in \mc{S}^{\text{\tiny GNE}}_N(\x)$ it holds that 
\begin{align}\label{eq:TurnpikeInequality}
Q_{\varepsilon} := \#\left\{k \in \bb{Z}_N \left|\, \left\| \begin{smallmatrix}x^*_k- x_s\\ u^*_k - u_s \end{smallmatrix} \right\|\leq \varepsilon\right\} \geq N - \frac{C}{\alpha(\varepsilon)}\right.
\end{align}
for some $\alpha \in \mc K$ and where $\#$ refers to the cardinality.\eDef
\end{dfn}
The turnpike property intuitively quantifies the number of steps a GNE trajectory spends within $\mc{B}_\varepsilon(x_s)$, an $\varepsilon$-ball around the steady-state GNE.

Strict dissipativity of the GNEP with respect to $(x_s, u_s)$ implies the turnpike property for the input and state trajectories resulting from~\eqref{eq:MPCPerAgent}. The result is inspired by~\cite[Thm. 5.3]{gruene2013economic}. Its proof can be found in \cite{hall2025system}. 

\begin{thm}[Strict dissipativity $\Rightarrow$ turnpike \cite{hall2025system}]\label{thm:DissToTurnpike}~\\
Consider the GNEP~\eqref{eq:MPCPerAgent} and let Assumptions~\ref{ass:cheap} and \ref{ass:PoA} hold. Suppose that the GNEP~\eqref{eq:MPCPerAgent} is strictly dissipative with respect to $(x_s,u_s)$ in the sense of Definition~\ref{dfn:StrictDiss} and the storage is bounded on $\bb X_N(\bb X_0)$. Then all GNEP solutions exhibit the (measure) turnpike property at $(x_s,u_s)$. \eDef
\end{thm}%

\section{Closed-loop Analysis of RHGs}

We make the following set of assumptions to derive our closed-loop stability results of the nonlinear RHG scheme.

\begin{ass}[Continuity and compactness]\label{ass:Continuity}
The constraint set $\bb{Z}_N$ is compact, and $\forall v \in \mc{V}$ the functions $\ell^v$,  $f$, and $\Lambda$ are continuous and $\ell$ is Lipschitz continuous with constant $L_\ell$.
\end{ass}

\begin{ass}[Local controllability at $(x_s, u_s)$]\label{ass:LocalControllability}
The Jacobian linearization of \eqref{eq:Dynamics} 
\[
A:= \dfrac{\partial f}{\partial x}\big|_{\left[\begin{smallmatrix}
    x \\u
\end{smallmatrix}\right] = \left[\begin{smallmatrix}
   x_s \\u_s
\end{smallmatrix}\right]}, \qquad 
B:= \dfrac{\partial f}{\partial u}\big|_{\left[\begin{smallmatrix}
    x \\u
\end{smallmatrix}\right] = \left[\begin{smallmatrix}
   x_s \\u_s
\end{smallmatrix}\right]}
\]
is controllable and $g(x_s, u_s^v, u_s^{-v})< 0$ , $h^v(u_s^v)< 0, \forall v \in \mc{V}$, i.e., $(x_s, u_s) \in \inte\bb Z$ . \eDef
\end{ass}

\begin{ass}[Finite time controllability into $\mathcal{B}_\varepsilon(x_s)$]\label{ass:FiniteControllability}
For any $\varepsilon > 0$  there exists $\hat N \in \bb{N}$ such that for each $\x \in \bb{X}_0$ there is $k \leq \hat N$ and a feasible pair $(x,u) \in \mc{Z}_N(\x)$ with
\[
x_k(\x) \in \mathcal{B}_\varepsilon(x_s).\vspace{-8mm}
\]
\vspace{2mm}\eDef
\end{ass}
\begin{rmk}[The turnpike property and controllability]

Observe that Theorem~\ref{thm:DissToTurnpike} implies that for each $\x \in \bb{X}_0$ and $\varepsilon > 0$, there exist $N(\varepsilon) \in \bb{N}$ and $k \in \bb{Z}_{N+1}$ such that $x^*_k(\x) \in \mathcal{B}_\varepsilon(x_s)$.
Thus, the result of Theorem~\ref{thm:DissToTurnpike} implies finite time controllability into $\mathcal{B}_\varepsilon(x_s)$. That is Assumption~\ref{ass:cheap} combined with strict dissipativity implies Assumption~\ref{ass:FiniteControllability}.
Conversely, combining Assumptions~\ref{ass:LocalControllability} and~\ref{ass:FiniteControllability} gives Assumption~\ref{ass:cheap}, where $\delta>0$ bounds the cost it takes to steer the initial condition to $x_s$.\eDef
\end{rmk}


The following technical results prepare the statement of our main closed-loop results. 

\begin{lmm}[Recursive feasibility of GNEPs]\label{lem:recursiveFeasibility}
    Suppose that the GNEP~\eqref{eq:MPCPerAgent} is strictly dissipative with respect to $(x_s,u_s)$ in the sense of Definition~\ref{dfn:StrictDiss} and for all states occurring along closed-loop RHG trajectories the storage $\Lambda(x)$ is bounded. Moreover, let the GNEP satisfy Assumptions~\ref{ass:PoA}--\ref{ass:FiniteControllability}.
    Then, there exists a finite horizon $\hat N$, such that for all $N\geq \hat N$, the GNEP is recursively feasible, i.e., if \eqref{eq:MPCPerAgent} is feasible for $\x_t$ then it is also feasible for $\x_{t+1}$.\eDef
\end{lmm}
\begin{proof}
  As a preparatory step, notice that Assumptions~\ref{ass:PoA}--\ref{ass:FiniteControllability} ensure that Theorem \ref{thm:DissToTurnpike} holds, i.e., the open-loop GNEP predictions computed at $\x_t \in \bb X_0$ exhibit the measure turnpike property as per Definition \ref{dfn:GameTurnpike}. Observe that due to  the turnpike property, for any $\varepsilon>0$,  there exists a finite GNEP horizon $N(\varepsilon)$ such that there exists $\tau \in \{0, \dots, N(\varepsilon)\}$ with $x^*_{\tau|t}(\x_t) \in \mathcal B_\varepsilon(x_s)$. In other words, the turnpike phenomenon implies that at some time point $\tau$ the trajectory will be $\varepsilon$-close to $x_s$ and we can choose $\varepsilon>0$ arbitrarily small. We remark that despite there being possibly infinitely many GNEs, the assumption of strict dissipativity combined with the deterministic selection $\kappa$ in \eqref{eq:FeedbackLaw} ensures that all GNE pairs exhibit dissipativity with respect to the same $(x_s, u_s)$.

    Suppose that at time step $t$, the GNEP \eqref{eq:MPCPerAgent} is feasible with initial condition $\x_t\in \bb X_0$ and that the computed game pair is $ (x^*_{\cdot|t}(\x_t) , u^*_{\cdot|t}(\x_t))$. 

    We consider the following candidate input sequence for the GNEP solved at time $t+1$.
    \begin{equation}\label{eq:feasibleInput}
    \tilde u_{\cdot|t+1} =\begin{cases}
        \begin{array}{ll}
          u^*_{k|t}, &k \in \{1, \dots, \tau\}\\
          \bar u_k, &k \in \{0, \dots,  n-1\}\\
          u^*_{k|t}, &k \in \{\tau+n, \dots, N-1\} 
        \end{array}
    \end{cases}
    \end{equation}
    which relies on the GNEP input $u^*_{\cdot|t}$ and on the n-step input $ \bar u \in \mc{Z}_n(x^*_{\tau+1|t})$. The latter has to satisfy
    $\bar x_n(\bar u, x^*_{\tau+1|t}) = x^*_{\tau+n|t}$. That is, the trajectory starting at $x^*_{\tau+1|t}$, driven by $\bar u$ reaches  $x^*_{\tau+n|t}$ after $n$ steps, as displayed in Figure~\ref{fig:Turnpike_sketch}. Notice that local controllability 
    (Assumption~\ref{ass:LocalControllability}) ensures that there exists $\varepsilon>0$ such that
    if $x^*_{\tau+1|t} \in \mathcal B_\varepsilon(x_s)$ there exists $ \bar u \in \mc{Z}_n(x^*_{\tau+1|t})$ with $\bar x_n(\bar u, x^*_{\tau+1|t}) = x^*_{\tau+n|t}$. For simplicity, we suppose that $x^*_{\tau+n|t}\in \mathcal B_\varepsilon(x_s)$. In case $x^*_{\tau+n|t}\not\in \mathcal B_\varepsilon(x_s)$, then the turnpike property ensures that (for a sufficiently long but finite horizon $N(\varepsilon)$) there exists some $m>n, m<N$ such that $x^*_{\tau+m|t}\in \mathcal B_\varepsilon(x_s)$.
    
    Assumption~\ref{ass:LocalControllability} further implies that $\bar x_j(\bar u, x^*_{\tau+1|t}) \in  \mc B_{\rho(\varepsilon)}(x_s) \supset \mc B_\varepsilon(x_s)$ for $j=0,\dots,n$. Here, $\rho(\varepsilon)$ is the size of the closed ball in which the finitely long trajectory pair $(\bar x, \bar u)$ evolves.\footnote{Formally, we have, 
    for all $k\in\bb Z_{n+1}$, that $(\bar x_k, \bar u_k) \in \mc B_{\rho(\varepsilon)}(x_s, u_s)$ $\Rightarrow$ $\bar x_k \in \mc B_{\rho(\varepsilon)}(x_s)$.} By the local continuity of the linearization error,\footnote{That is, the error between the nonlinear dynamics and the linearized one varies continuously with $\varepsilon>0$ and goes to $0$ as $\varepsilon\to 0$.} we have that $\rho(\varepsilon) \to 0$ as $\varepsilon \to 0$, i.e., the size of the ball bounding the trajectory pair $(\bar x, \bar u)$ shrinks continuously with $\varepsilon \geq \|x^*_{\tau+1|t} - x_s\|$.
    \begin{figure}
        \centering
    \includegraphics[width=\linewidth]{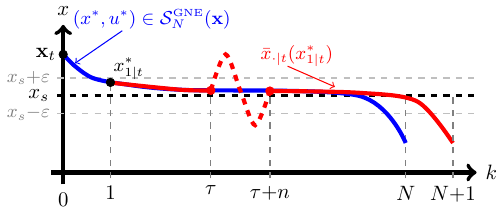}
        \caption{A schematic showing the turnpike property of $(x^*,u^*)$ with respect to $x_s$ and portraying the candidate state sequence~\eqref{eq:feasibleInput} employed to prove Lemma~\ref{lem:recursiveFeasibility}.}
        \label{fig:Turnpike_sketch}
    \end{figure}

    In general, the next closed-loop state $\x_{t+1}$ resulting from the RHG feedback law \eqref{eq:FeedbackLaw} does not need to be in $\bb X_0$. However, from $\x_{t+1} = x^*_{1|t}(\x_t)$ we may apply the first $\tau$ steps of  $\tilde u_{k|t+1}$ to steer the state $\varepsilon$-close to the turnpike state $x_s$, with $\varepsilon>0$ arbitrarily small. Due to Assumption \ref{ass:LocalControllability}, we have that there exists some $\bar\varepsilon(\varepsilon) \geq \varepsilon$ such that (i) for all $\x \in\mathcal B_{\bar\varepsilon(\varepsilon)}(x_s)$ we have $\mc{Z}_n(\x) \not =\emptyset$ and (ii) from all $\x \in\mathcal B_{\varepsilon}(x_s)\subset \mathcal B_{\bar\varepsilon(\varepsilon)}(x_s)$ the turnpike state $x_s$ can be reached in $n$ steps without violating any input or state constraint and without leaving $\mathcal B_{\bar\varepsilon(\varepsilon)}(x_s)$. 
    This means that for sufficiently long horizons $N \geq N(\varepsilon)$, cheap reachability (Assumption \ref{ass:cheap}) holds for the next initial condition $\x_{t+1}$ with the same bounds. Thus at $\x_{t+1}$ the conditions of Theorem~\ref{thm:DissToTurnpike} hold again for the GNEP solutions computed at $\x_{t+1}$. Thus also at $t+1$ the turnpike property holds. 
\end{proof}

Next, similar to the sampled-data OCP results \cite{faulwasser15approx}, we propose the candidate Lyapunov function 
\begin{equation}\label{eq:LyapCand}
  W(x) := V_N^*(x) + \Lambda(x) 
\end{equation}
and analyze its properties. 

\begin{lmm}[Properties of W(x)]\label{lem:W}
    Suppose that the GNEP~\eqref{eq:MPCPerAgent} is strictly dissipative with respect to $(x_s,u_s)$ in the sense of Definition~\ref{dfn:StrictDiss} and for all states occurring along closed-loop RHG trajectories  the storage is bounded. Moreover, let the GNEP satisfy Assumptions~\ref{ass:PoA}--\ref{ass:FiniteControllability}.
    Then,  $W$ from \eqref{eq:LyapCand} satisfies the following:
    \begin{itemize}
        \item[(i)] It holds that
         $ \alpha_\ell(\|\x-x_s\|) \leq W(\x)$.
        \item[(ii)] There exists $\tilde \rho >0$ such that
    \begin{equation*}
      W(\x_{t+1})- W(\x_t) < 0
    \end{equation*}
    holds for all $\x_t \not\in  \mathcal B_{\tilde \rho}(x_s)$. \eDef
    \end{itemize}
\end{lmm}
The decrease condition does not hold inside the ball $\mathcal{B}_{\tilde\rho}(x_s)$ since the dissipation decrease $\alpha_\ell(\|x-x_s\|)$ is dominated by the error induced by the turnpike tail. This is well-known for finite-horizon problems without terminal ingredients.

\begin{proof}
\emph{Part (i):} We first consider the following offset on the per-agent stage cost
\begin{equation}\label{eq:barell}
    \bar\ell^v(x,u^v, u^{-v}):= \ell^v(x,u^v, u^{-v}) - \ell^v(x_s, u^v_s, u^{-v}_s),  
\end{equation}
which ensures that 
$
\bar\ell(x,u):= \sum_{v\in \mathcal V} \bar\ell^v(x,u^v, u^{-v})$
satisfies $\bar\ell(x_s,u_s)=0$. Observe, that the GNE solution sets of GNEP \eqref{eq:MPCPerAgent} with $\ell^v$ and $\bar\ell^v$ remain unchanged. Henceforth we assume  that such an offset has been applied and thus we set without loss of generality $\ell(x_s, u_s) = 0$.

Moreover, observe that if $\Lambda$ satisfies the strict dissipation inequality \eqref{eq:GsDI}, then $\Lambda + c$, $c\in \R$ does so as well. Hence boundedness of the storage function $\Lambda$ and compactness of the state constraints (Assumption~\ref{ass:Continuity}) imply that we can suppose without loss of generality that $\Lambda(x) \geq 0$ for all states occurring along closed-loop RHG trajectories.

Consider $W$ from \eqref{eq:LyapCand}, the strict dissipation inequality with $\Lambda(x) \geq 0$ and $\ell(x_s, u_s) = 0$ gives
\[
\Lambda(f(x^*_{k|t},u^*_{k|t})) -\Lambda(x^*_{k|t}) \leq  \ell^*_k - \alpha^*_{\ell,k}
\]
where $\ell^*_k$ and $\alpha^*_{\ell,k}$ are shorthands for the $k$ dependent terms on the right hand side of \eqref{eq:GsDI}. Now using a telescopic sum and evaluating \eqref{eq:GsDI} from $k=0$ to $k = N-1$
gives
\begin{multline*}
    \Lambda(f(x^*_{N-1|t},u^*_{N-1|t}))+\sum_{k=0}^{N-1} \alpha^*_{\ell,k} \\\leq \Lambda(\x_t)+\sum_{k=0}^{N-1} \ell^*_k= V_N^*(\x_t)+ \Lambda(\x_t). 
\end{multline*}
The term on the left side can be bounded from below by
\[
\alpha^*_{\ell,0} \leq  \Lambda(f(x^*_{N-1|t},u^*_{N-1|t}))+\sum_{k=0}^{N-1} \alpha^*_{\ell,k}
\]
and thus with $\alpha^*_{\ell,0} = \alpha_{\ell}\left(\left\|
\begin{smallmatrix}\x_t- x_s\\ u^*_{0|t}- u_s \end{smallmatrix}\right\|\right)$ it follows
\[
\alpha_\ell(\|\x_t - x_s\|) \leq \alpha_{\ell}\left(\left\|
\begin{smallmatrix}\x_t- x_s\\ u^*_{0|t}- u_s \end{smallmatrix}\right\|\right)\leq W(\x_t)
\]
which shows Part (i).

\emph{Part (ii):}
Assumption \ref{ass:PoA} (bounded price of anarchy) gives
\begin{multline*}
     \Delta W:= W(\x_{t+1})- W(\x_t) \\= V_N^*(\x_{t+1}) + \Lambda(\x_{t+1}) - V_N^*(\x_t)-  \Lambda(\x_t)\\
    \leq V_N^\diamond(\x_{t+1})+\bar P + \Lambda(\x_{t+1}) - V_N^*(\x_t)-  \Lambda(\x_t).
\end{multline*}

Next, we use the fact that the OCP value function $V_N^\diamond(\x_{t+1})$ can be bounded from above by the performance of the feasible input 
$ \tilde u_{\cdot|t+1}$ defined in~\eqref{eq:feasibleInput}, i.e., we have 
\[
J(\tilde x_{\cdot|t+1}, \tilde u_{\cdot|t+1}) = \sum_{k=0}^{N-1}\ell( \tilde x_{k|t+1},  \tilde u_{k|t+1}) \geq V_N^\diamond(\x_{t+1}),
\]
where $\tilde x_{\cdot|t+1}$ is the trajectory generated by  $\tilde u$ and initial condition $\x_{t+1}$.
Hence
\begin{align}\label{eq:DeltaW}
    \Delta W \leq J_{t+1}  - V_N^*(\x_t) +\bar P + \Lambda(\x_{t+1})-  \Lambda(\x_t),
\end{align}
where we use the shorthand $J_{t+1}:=J( \tilde x_{\cdot|t+1}, \tilde u_{\cdot|t+1})$.
Due to \eqref{eq:feasibleInput} and due to the absence of plant-model mismatch in the RHG, the functional $J_{t+1}$   contains trajectory parts of $ (x^*_{\cdot|t} , u^*_{\cdot|t})$. For the sake of readability, we suppress the dependence of this game pair on the initial condition $\x_t$.  In particular, the identity
\[
\ell(x^*_{k|t} , u^*_{k|t}) = \ell( \tilde x_{j|t+1},  \tilde u_{j|t+1})
\]
holds for all $k=j+1$, $j\in \{0,\dots,  \tau -1\}$ and for $k=j$, $j\in \{\tau+n, \dots,  N-1\}$.
Hence, we arrive at
\begin{multline}\label{eq:Bnd1}
    J_{t+1}- V_N^*(\x_t) = -\ell(x^*_{0|t} , u^*_{0|t}) + \ell(x_s, u_s)\\+ 
    \sum_{j=\tau}^{\tau+n-1}\ell( \tilde x_{j|t+1},  \tilde u_{j|t+1}) - \ell(x_s, u_s) \\-\sum_{k=\tau+1}^{\tau+n}\ell(x^*_{k|t} , u^*_{k|t})- \ell(x_s, u_s),  
\end{multline}
where we added $\sum_{k=0}^{N-1} \ell(x_s, u_s) - \sum_{k=0}^{N-1} \ell(x_s, u_s) = 0$ using the negative sum in $J( \tilde x_{\cdot|t+1},\tilde u_{k|t+1})$ and the positive one in $- V_N^*(\x_t)$.
The second sum in \eqref{eq:Bnd1} admits the bound
\begin{equation*}
\sum_{j=\tau}^{\tau+n-1} \left( \ell( \tilde x_{j|t+1},  \tilde u_{j|t+1}) - \ell(x_s, u_s)\right) \leq nL_\ell \rho(\varepsilon),
\end{equation*}
where $n = \dim x$ and $L_\ell$ is the Lipschitz constant of $\ell$. Moreover, $\rho(\varepsilon) \geq \varepsilon$ is the radius of the ball bounding the $n$-step trajectory pair $(\bar x, \bar u)$ introduced in the proof of Lemma~\ref{lem:recursiveFeasibility}. Likewise, we have
\begin{equation*}  
 -\sum_{k=\tau+1}^{\tau+n} \left(\ell(x^*_{k|t} , u^*_{k|t}) - \ell(x_s, u_s) \right) \leq nL_\ell \Delta(\varepsilon),
\end{equation*}
where $\Delta(\varepsilon) = \max_{k \in \{\tau+1, \dots,\tau+n \}} \|(x^*_{k|t} , u^*_{k|t}) - (x_s, u_s)\|$.
We obtain further
\begin{subequations} \label{eq:Bnd2}
\begin{multline}
     J_{t+1}- V_N^*(\x_t) \leq \\-\ell(x^*_{0|t} , u^*_{0|t}) + \ell(x_s, u_s) + nL_\ell ( \rho(\varepsilon)+\Delta(\varepsilon)).
\end{multline}
Next, we bound the term $ \Lambda(\x_{t+1})-  \Lambda(\x_t)$ using the strict dissipation inequality \eqref{eq:GsDI} and obtain
\begin{multline}
     \Lambda(\x_{t+1})-  \Lambda(\x_t) \leq -  \alpha_{\ell}\left(\left\|
\begin{smallmatrix}x^*_{0|t}- x_s\\ u^*_{0|t} - u_s \end{smallmatrix}
\right\|\right) \\+ \ell(x^*_{0|t} , u^*_{0|t}) - \ell(x_s, u_s).
\end{multline}
\end{subequations}
Using \eqref{eq:Bnd2} in \eqref{eq:DeltaW} we obtain that
\begin{multline} \label{eq:BndDeltaW}
     \Delta W \leq  -  \alpha_{\ell}(\|x^*_{0|t}- x_s \|) + \bar P+ nL_\ell ( \rho(\varepsilon)+\Delta(\varepsilon)).
\end{multline}
Hence whenever 
\[\|x^*_{0|t}- x_s \|> \alpha^{-1}_\ell(\bar P+ nL_\ell ( \rho(\varepsilon)+\Delta(\varepsilon)))\]
we have that $ \Delta W < 0$    and thus also
\[
 W(\x_{t+1})- W(\x_t) < 0, \qquad\forall \x_t \not\in  \mathcal B_{\tilde \rho}(x_s)
\]
with $\tilde\rho := \alpha^{-1}_\ell \left(\bar P+ nL_\ell ( \rho(\varepsilon)+\Delta(\varepsilon))\right)$.
This shows Part~(ii) and finishes the proof.
\end{proof}

We now come to our main result. 
 \begin{figure}
        \centering
    \includegraphics[width=\linewidth]{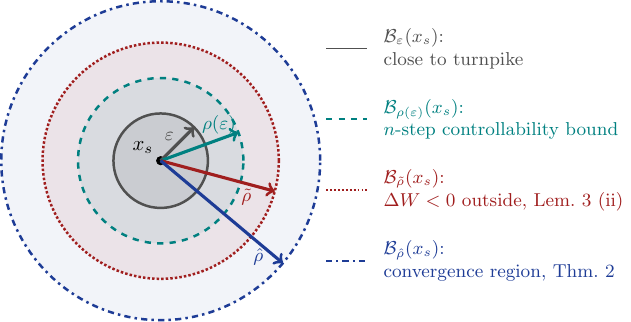}
        \caption{A schematic of the balls $\mc{B}$ around $x_s$ used in Lemma~\ref{lem:recursiveFeasibility}, Lemma~\ref{lem:W}, and Theorem~\ref{thm:practicalConv}.}
        \label{fig:Turnpike_regions}
    \end{figure}
Recall the standard notion of the distance of a point $\x\in\R^{n_x}$ to a compact set $\bb X$ which is
$\displaystyle 
 \dist(\x, \bb X):= \min_{ \tilde x\in \bb X} \|\tilde x - \x\|$.\footnote{In case of open sets one swaps the $\min$ with $\inf$.}

\begin{thm}[Practical asymptotic convergence]\label{thm:practicalConv}
Suppose that the GNEP~\eqref{eq:MPCPerAgent} is strictly dissipative with respect to $(x_s,u_s)$ in the sense of Definition~\ref{dfn:StrictDiss}, and for all states occurring along closed-loop RHG trajectories the storage is bounded. Moreover, let the GNEP satisfy Assumptions~\ref{ass:PoA}--\ref{ass:FiniteControllability}, and let there exist a finite $a >0$ such that 
\[W(x) \leq a, \quad \forall x \in \mathcal B_{\tilde \rho}(x_s)\]
with $\tilde \rho = \alpha^{-1}_\ell \left( \bar P + nL_\ell ( \rho(\varepsilon)+\Delta(\varepsilon))\right)$ holds. Then, there exists $\bar{N}\in \bb{N}$ such that for all $N\geq \bar N$ the closed loop RHG satisfies
\[
\lim_{t\to\infty} \dist(\x_t, \mathcal{B}_{\hat\rho}(x_s)) = 0
\]
for some $\hat\rho \geq \tilde \rho$.
\end{thm}
\begin{proof}
    The proof proceeds in three main steps: First we bound the difference between $W$ and $\alpha_\ell$ on $ \mathcal B_{\tilde \rho}(x_s)$. Second we analyze the asymptotics for $\x_t \not \in  \mathcal B_{\tilde \rho}(x_s)$, and third we consider the case $\x_{t} \in  \mathcal B_{\tilde \rho}(x_s)$.
    
    As a preparatory step observe that the setting of the theorem includes the conditions used for Lemmas \ref{lem:recursiveFeasibility} and \ref{lem:W}.

    \emph{Step 1:} Lemma \ref{lem:W}, Part (i)  gives
    \[
    \alpha_\ell(\|x-x_s\|) \leq W(x)\leq a, \quad \forall x \in \mathcal B_{\tilde \rho}(x_s)
    \]
    and hence $W(x) -  \alpha_\ell(\|x-x_s\|) \leq a -  \alpha_\ell(\|x-x_s\|)$.
    Thus \[
    \sup_{x \in \mathcal B_{\tilde \rho}(x_s)}\hspace{-2mm} W(x) -  \alpha_\ell(\|x-x_s\|) \leq \max_{x \in \mathcal B_{\tilde \rho}(x_s)} \hspace{-2mm}a -  \alpha_\ell(\|x-x_s\|)= a.
    \]

    \emph{Step 2:} Whenever $\x_t \not \in  \mathcal B_{\tilde \rho}(x_s)$, we have from Part (ii) of Lemma \ref{lem:W} that $W(\x_{t+1}) - W(\x_t) < 0$ and hence $W$ decays. This means that for $\x_t \not \in  \mathcal B_{\tilde \rho}(x_s)$ the next state $\x_{t+1}$ will be closer to $\mathcal B_{\tilde \rho}(x_s)$. However, it is not directly clear if once the state enters $\mathcal B_{\tilde \rho}(x_s)$ it will remain inside this set. 
    
    \emph{Step 3:} Whenever $\x_t \in  \mathcal B_{\tilde \rho}(x_s)$ Lemma~\ref{lem:W}~(ii) cannot be applied.  However, the bound \eqref{eq:BndDeltaW} still holds and gives 
    \[
    W(\x_{t+1}) - W(\x_t) \leq -  \alpha_{\ell}(\|\x_t- x_s \|) + \bar P+ nL_\ell ( \rho(\varepsilon)+\Delta(\varepsilon))
    \]
    and thus
    \begin{align*}
         W(\x_{t+1}) &\leq W(\x_t)  -  \alpha_{\ell}(\|\x_t- x_s \|) + \bar P+ nL_\ell ( \rho(\varepsilon)+\Delta(\varepsilon))\\
           W(\x_{t+1}) &\geq \alpha_{\ell}(\|\x_{t+1}- x_s \|)
    \end{align*}
   where the second inequality stems from  Part (i) of  Lemma \ref{lem:W}.
   Using Step 1 again, and combining both inequalities gives
   \[
    \alpha_{\ell}(\|\x_{t+1}- x_s \|) \leq  W(\x_{t+1}) \leq a  + \bar P+ nL_\ell (  \rho(\varepsilon)+\Delta(\varepsilon)).
   \]
   Now apply $\alpha_\ell^{-1}$ to obtain 
   \[
   \|\x_{t+1}- x_s \| \leq \alpha_\ell^{-1}\left(a  + \bar P+ nL_\ell ( \rho(\varepsilon)+\Delta(\varepsilon))\right).
   \]
   That is, if $\x_t \in \mathcal B_{\tilde \rho}(x_s)$ then we have that $\x_{t+1} \in  \mathcal B_{\hat\rho}(x_s)$ with $\hat\rho := \alpha_\ell^{-1}\left(a  + \bar P+  nL_\ell (  \rho(\varepsilon)+\Delta(\varepsilon))\right)$.
   The relation $\hat\rho \geq \tilde \rho = \alpha_\ell^{-1}\left(\bar P + nL_\ell ( \rho(\varepsilon)+\Delta(\varepsilon))\right)$ follows from the fact that the inverse $\alpha_\ell^{-1}$ on a compact domain is of class $ \mathcal K$ if $\alpha_\ell \in \mathcal K$.\footnote{
   If a compact domain of $\alpha_\ell$ is considered, the inverse $\alpha_\ell^{-1}$ is of class $\mc K$, see \cite{Kellett14}. Assumption~\ref{ass:Continuity} ensures the compactness of the domain of $\alpha_\ell$. Alternatively, one could define strict dissipativity in~\eqref{eq:GsDI} with $\alpha_\ell \in \mc K_\infty$.} 
   Now, $\forall \, \x_t \in  \mathcal B_{\hat\rho}(x_s) \setminus \mathcal B_{\tilde \rho}(x_s)$, the decay of $W(\x_t)$ from Step 2 holds. Hence the larger set $ \mathcal B_{\hat\rho}(x_s)$  is forward invariant for the closed-loop RHG solution. 
   This finishes the proof.
\end{proof}

The reader might ask how to go from practical asymptotic convergence to practical asymptotic stability and how to quantify the size of the neighborhood $ \mathcal B_{\hat\rho}(x_s)$. The next result comments on the first issue, while the second one is explored via numerical examples in Section~\ref{sec:Examples}.
\begin{prp}
Consider the setting and assumptions of Theorem \ref{thm:practicalConv}. If $\exists \; \alpha_{W} \in  \mathcal K$ such that 
$W(x) \leq \alpha_{ W}(\|x-x_s\|)$ for all $x \in \bb X \supset  \mathcal B_{\hat\rho}(x_s)$,
then \eqref{eq:FeedbackLaw} is locally practically asymptotically stable on the compact set $ \bb X$. \eDef
\end{prp}

The proof follows analogously to that of practical asymptotic stability for dissipativity-based NMPC and is omitted due to space limitations. However, the key open question, left for future work, is verifying the existence of an upper bound on $W$ in the RHG setting.

\section{Numerical studies} \label{sec:Examples}

\subsection*{Linear Quadratic with Coupled Dynamics}
We consider a linear-quadratic instance of the GNEP in~\eqref{eq:MPCPerAgent} with coupled LTI dynamics, coupled costs, and constraints
\begin{equation}\label{eq:GNEP_example}
\left\{
\begin{array}{r l}
\displaystyle \min_{u^v, x}  &\displaystyle \sum_{k= 0}^{N-1}  \displaystyle u_k^v\left(\sum_{j\in \mc{V}} R^{v,j} u_k^{j}\right) + \|x_k- x^{\text{ref}}\|_{Q^v}^2\\
\textrm{s.t.} &  x_{k+1} = A x_k +  \displaystyle \sum_{j\in \mc{V}} B^j u_k^j ,  \hspace{0.1em}  k \in \bb{Z}_{N}\\
& -2\leq u_k^v \leq 2, \hspace{4.6em}
k \in \bb{Z}_{N} \\
& -2\leq \displaystyle \sum_{j\in\mc{V}} u_k^j\leq 2, \hspace{2.9em}
k \in \bb{Z}_{N}
\\ 
&  - 1\leq x_k\leq 1,   x_0 =1, \hspace{1.3em} k \in \bb{Z}_{N+1}.
\end{array}
\right.
\end{equation}
 with $ v \in \{1,2\}$. The parameter values are $A = 1.5$, $B^1 = 1$, $B^2 = 2$, $R^{1,1} = R^{1,2} = 4$, $R^{2,2} = R^{2,1} = 5$ and state weights $Q^1 = 1, Q^2 = 2$. The reference state is $x^{\text{ref}}=0.3$. We solve for GNEs of~\eqref{eq:GNEP_example} using a regularized Fischer--Burmeister method~\cite{liaomcpherson2019regularized}. We show the resulting trajectories in Figure~\ref{fig:ClosedLoopComparison}(a), where the open-loop predictions exhibit the characteristic turnpike leaving arc. In Figure~\ref{fig:ClosedLoopComparison}(b) this arc is fully suppressed using a linear end penalty, the design of which is described in~\cite{hall2025system}. 
 In Figure~\ref{fig:HorizonConvergence}(a) we demonstrate that the convergence region decreases with increasing horizon length and that applying an end penalty yields exact convergence in Figure \ref{fig:HorizonConvergence}(b).

\begin{figure}[t]
    \centering
    \begin{subfigure}[b]{\columnwidth}
        \centering
        \includegraphics[width=\columnwidth]{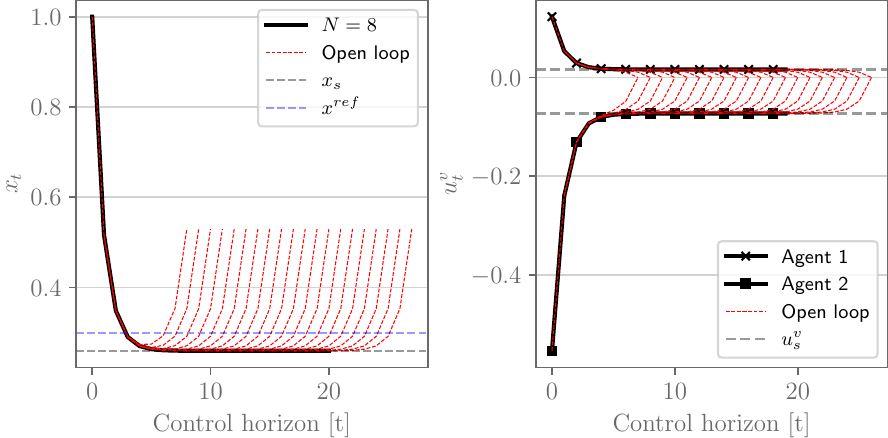}
        \caption{Without terminal penalty.}
        \label{fig:ClosedLoopState}
    \end{subfigure}
    \begin{subfigure}[b]{\columnwidth}
        \centering
        \includegraphics[width=\columnwidth]{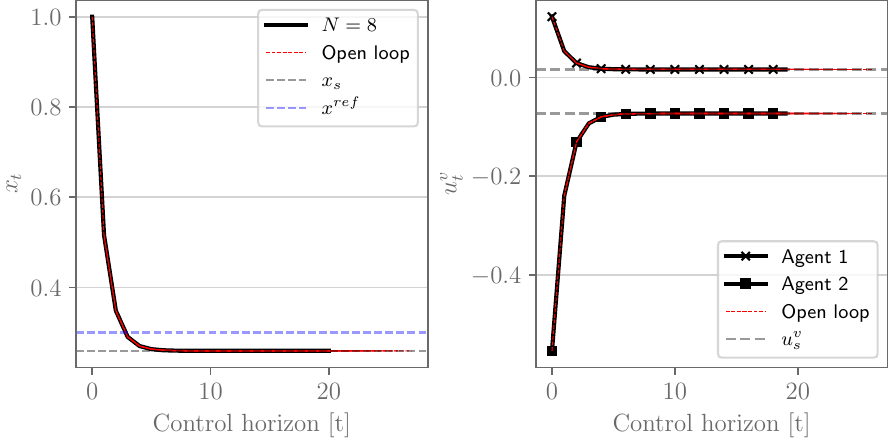}
        \caption{With exact terminal penalty $J^v_N(x,u) + \lambda_s x_N$.}
        \label{fig:ClosedLoopStateExact}
    \end{subfigure}
    \caption{Comparison of closed-loop GNE trajectories of~\eqref{eq:GNEP_example} with open-loop trajectories (red dashed lines) and $N=8$.}
    \label{fig:ClosedLoopComparison}
\end{figure}

\begin{figure}[t]
    \centering
    \begin{subfigure}[b]{0.49\columnwidth}
        \centering
\includegraphics[width=\columnwidth]{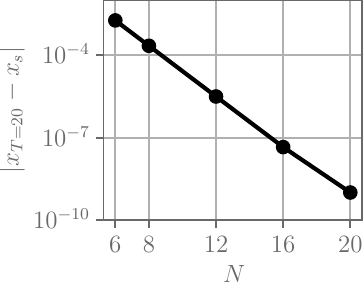}
\caption{No penalty}
 \end{subfigure}%
    \begin{subfigure}[b]{0.49\columnwidth}
        \centering
        \includegraphics[width=\columnwidth]{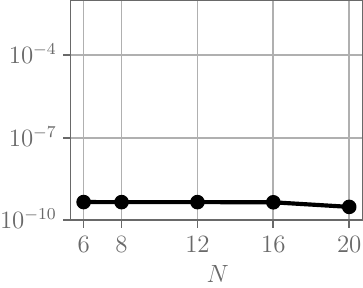}
        \caption{With penalty $\lambda_s\, x_N$} 
            \label{fig:HorizonConvergencePenalty}
    \end{subfigure}
    \caption{Convergence to $x_s$ of closed-loop GNE trajectories of~\eqref{eq:GNEP_example} after $T=20$ control steps. }
     \label{fig:HorizonConvergence}
\end{figure}




\subsection*{Nonlinear Economic Growth Model}

We adapt the nonlinear example from \cite{gruene2013economic, gruene2014asymptotic} to a game-theoretic setting with two agents and local dynamics as follows $\forall v \in \mathcal{V}=\{1,2\}$: 

\begin{equation}\label{eq:GNEP_economic_growth}
\left\{
\begin{array}{r l}
\displaystyle \min_{u^v, x}  &\displaystyle \sum_{k= 0}^{N-1}  \displaystyle 
-\ln\Big(q^v(x_k^v)^{\alpha^v} -r^v  u_k^v \sum_{j\in \mc{V}} u_k^j\Big) \\
\textrm{s.t.} & \displaystyle  x_{k+1}^v =u_k^v  \hspace{4.6em}  k \in \bb{Z}_{N}\\
& 0.1\leq \displaystyle \sum_{j\in\mc{V}} u_k^j\leq 5, \hspace{1.3em} k \in \bb{Z}_{N} \\ 
& 0\leq \displaystyle x_k^v\leq 10, \hspace{0.3em}
\hspace{2.9em} k \in \bb{Z}_{N+1} \\ 
 & x_0^v =1. 
\end{array}
\right.
\end{equation}
where $q^1 = 5$, $q^2= 4$ are the productivity rates, $r^1 = 1$, $r^2= 1.5$ are the cost coefficients for investment interaction, and $\alpha^1 = 0.3, \alpha^2=  0.2$ represent the capital shares. We solve the nonlinear GNEP using the~\textit{NashOpt} library~\cite{bemporad2025nashopt}. In Figure~\ref{fig:ClosedLoopEconGrowth} the closed-loop RHG trajectory is plotted for a prediction horizon of $N= 12$ and we clearly see the convergence of the closed-loop to the steady-state GNE $x_s$ as well as the turnpike property in the open-loop predictions (red dashed lines). Figure~\ref{fig:HorizonConvergenceEconGrowth} demonstrates the exponential convergence of the closed-loop RHG trajectory to the steady-state GNE as a function of the prediction horizon.

\begin{figure}[t]
\centering
\includegraphics[width=\columnwidth]{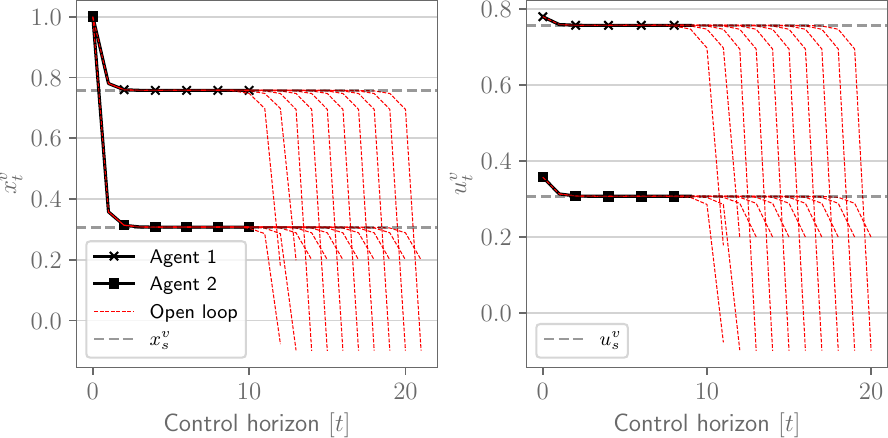}
\caption{Closed-loop trajectories (black-solid) and open-loop predictions (red-dashed) of~\eqref{eq:GNEP_economic_growth} with $N=12$.} \label{fig:ClosedLoopEconGrowth}
\end{figure}

\begin{figure}[t]
\centering
\includegraphics[width=0.5\columnwidth]{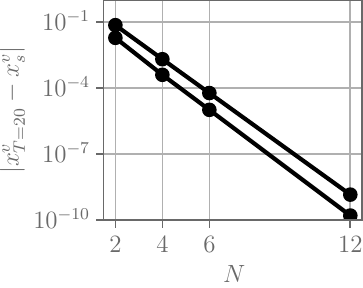}
\caption{Convergence of closed-loop trajectories of~\eqref{eq:GNEP_economic_growth} to $x_s$ as a function of $N$.} \label{fig:HorizonConvergenceEconGrowth}
\end{figure}

\section{Conclusion}
This paper has established closed-loop stability results for nonlinear Receding Horizon Games which are based on dynamic GNE trajectories applied in a receding horizon fashion. Under a strict dissipativity assumption, we show that the turnpike property with respect to the steady-state GNE ensures recursive feasibility of the RHG feedback law. We then construct a Lyapunov candidate $W(x)$ as the sum of a group performance measure and a storage function, and prove practical asymptotic convergence to a neighborhood of the steady-state GNE. When $W(x)$ additionally admits a local $\mathcal{K}$ upper bound in $\|x - x_s\|$, this strengthens to practical asymptotic stability. To the best of our knowledge, these are the first closed-loop stability results in a nonlinear RHG setting without terminal penalties or constraints.

Future work will investigate conditions under which the local $\mathcal{K}$ upper bound on $W(x)$ holds, and whether linear terminal penalties constructed from steady state Lagrange multipliers can yield exact asymptotic stability.

\addtolength{\textheight}{-3cm}   



\bibliography{Dissipativity_Turnpikes_GNEPs}

@Article{diehl2011lyapunov,
  author     = {Moritz Diehl and Rishi Amrit and James B. Rawlings},
  journal    = {{IEEE} Transactions on Automatic Control},
  title      = {{A Lyapunov Function for Economic Optimizing Model Predictive Control}},
  year       = {2011},
  month      = mar,
  number     = {3},
  pages      = {703--707},
  volume     = {56},
  doi        = {10.1109/tac.2010.2101291},
  publisher  = {Institute of Electrical and Electronics Engineers ({IEEE})},
  readstatus = {read},
  relevance  = {relevant},
}

@Article{wang2021game,
  author    = {Mingyu Wang and Zijian Wang and John Talbot and J. Christian Gerdes and Mac Schwager},
  journal   = {{IEEE} Transactions on Robotics},
  title     = {Game-Theoretic Planning for Self-Driving Cars in Multivehicle Competitive Scenarios},
  year      = {2021},
  pages     = {1--13},
  doi       = {10.1109/tro.2020.3047521},
  publisher = {Institute of Electrical and Electronics Engineers ({IEEE})},
}

@Article{mignoni2023distributed,
  author    = {Mignoni, Nicola and Carli, Raffaele and Dotoli, Mariagrazia},
  journal   = {IEEE Transactions on Control Systems Technology},
  title     = {Distributed Noncooperative MPC for Energy Scheduling of Charging and Trading Electric Vehicles in Energy Communities},
  year      = {2023},
  publisher = {IEEE},
}

@InProceedings{paola2018distributed,
  author     = {Antonio De Paola and Filiberto Fele and David Angeli and Goran Strbac},
  booktitle  = {2018 {IEEE} Conference on Decision and Control ({CDC})},
  title      = {Distributed Coordination of Price-Responsive Electric Loads: A Receding Horizon Approach},
  year       = {2018},
  month      = {dec},
  publisher  = {{IEEE}},
  date       = {2018-12},
  doi        = {10.1109/cdc.2018.8619024},
  readstatus = {skimmed},
}

@Article{gu2008differential,
  author    = {Dongbing Gu},
  journal   = {IEEE Trans. Control Syst. Technol.},
  title     = {A Differential Game Approach to Formation Control},
  year      = {2008},
  month     = jan,
  number    = {1},
  pages     = {85--93},
  volume    = {16},
  doi       = {10.1109/tcst.2007.899732},
  publisher = {Institute of Electrical and Electronics Engineers ({IEEE})},
}

@Article{liniger2020noncooperative,
  author     = {Alexander Liniger and John Lygeros},
  journal    = {{IEEE} Transactions on Control Systems Technology},
  title      = {A Noncooperative Game Approach to Autonomous Racing},
  year       = {2020},
  month      = may,
  number     = {3},
  pages      = {884--897},
  volume     = {28},
  doi        = {10.1109/tcst.2019.2895282},
  publisher  = {Institute of Electrical and Electronics Engineers ({IEEE})},
  readstatus = {skimmed},
}

@Misc{hall2025solving,
  author    = {Hall, Sophie and Bemporad, Alberto},
  title     = {Solving Multiparametric Generalized Nash Equilibrium Problems and Explicit Game-Theoretic Model Predictive Control},
  year      = {2025},
  copyright = {Creative Commons Attribution Non Commercial No Derivatives 4.0 International},
  doi       = {10.48550/ARXIV.2512.05505},
  publisher = {arXiv},
}

@Article{benenati2023optimal,
  author   = {Benenati, Emilio and Ananduta, Wicak and Grammatico, Sergio},
  journal  = {IEEE Transactions on Automatic Control},
  title    = {Optimal Selection and Tracking Of Generalized Nash Equilibria in Monotone Games},
  year     = {2023},
  number   = {12},
  pages    = {7644-7659},
  volume   = {68},
  doi      = {10.1109/TAC.2023.3288372},
}

@InProceedings{hall2022receding,
  author    = {Hall, Sophie and Belgioioso, Giuseppe and Liao-McPherson, Dominic and D{\"o}rfler, Florian},
  booktitle = {2022 IEEE 61st Conference on Decision and Control (CDC)},
  title     = {Receding Horizon Games with Coupling Constraints for Demand-Side Management},
  year      = {2022},
  pages     = {3795-3800},
  doi       = {10.1109/CDC51059.2022.9992497},
}

@Article{gruene2016relation,
  author    = {Gr{\"u}ne, Lars and M{\"u}ller, Matthias A.},
  journal   = {Systems \& Control Letters},
  title     = {On the relation between strict dissipativity and turnpike properties},
  year      = {2016},
  issn      = {0167-6911},
  month     = apr,
  pages     = {45--53},
  volume    = {90},
  doi       = {10.1016/j.sysconle.2016.01.003},
  publisher = {Elsevier BV},
}

@Article{hall2025stability,
  author   = {Hall, Sophie and Belgioioso, Giuseppe and D{\"o}rfler, Florian and Liao-McPherson, Dominic},
  journal  = {IEEE Transactions on Automatic Control},
  title    = {Stability Certificates for Receding Horizon Games},
  year     = {2025},
  pages    = {1-8},
  doi      = {10.1109/TAC.2025.3647314},
}

@article{faulwasser2018economic,
  title={Economic nonlinear model predictive control},
  author={Faulwasser, Timm and Gr{\"u}ne, Lars and M{\"u}ller, Matthias A and others},
  journal={Foundations and Trends{\textregistered} in Systems and Control},
  volume={5},
  number={1},
  pages={1--98},
  year={2018},
  publisher={Now Publishers, Inc.}
}

@Article{faulwasser2022turnpike,
  author    = {Faulwasser, Timm and Gr{\"u}ne, Lars},
  journal   = {Handbook of numerical analysis},
  title     = {Turnpike properties in optimal control: An overview of discrete-time and continuous-time results},
  year      = {2022},
  pages     = {367--400},
  volume    = {23},
  publisher = {Elsevier},
}

@Article{gruene2013economic,
  author    = {Gr{\"u}ne, Lars},
  journal   = {Automatica},
  title     = {Economic receding horizon control without terminal constraints},
  year      = {2013},
  month     = {mar},
  number    = {3},
  pages     = {725--734},
  volume    = {49},
  doi       = {10.1016/j.automatica.2012.12.003},
  publisher = {Elsevier {BV}},
}

@InCollection{facchinei2009nash,
  author    = {Francisco Facchinei and Jong Pang},
  booktitle = {Convex Optimization in Signal Processing and Communications},
  publisher = {Cambridge University Press},
  title     = {{Nash equilibria: the variational approach}},
  year      = {2009},
  chapter   = {12},
  editor    = {Daniel P. Palomar and Yonina C. Eldar},
  month     = {dec},
  pages     = {443--493},
  doi       = {10.1017/cbo9780511804458.013},
  place     = {Cambridge},
}

@Article{facchinei2009generalized,
  author    = {Facchinei, Francisco and Kanzow, Christian},
  journal   = {Annals of Operations Research},
  title     = {Generalized Nash Equilibrium Problems},
  year      = {2009},
  issn      = {1572-9338},
  month     = nov,
  number    = {1},
  pages     = {177--211},
  volume    = {175},
  doi       = {10.1007/s10479-009-0653-x},
  publisher = {Springer Science and Business Media LLC},
}

@Article{angeli2012average,
  author     = {Angeli, David and Amrit, Rishi and Rawlings, James B.},
  journal    = {IEEE Transactions on Automatic Control},
  title      = {On Average Performance and Stability of Economic Model Predictive Control},
  year       = {2012},
  number     = {7},
  pages      = {1615-1626},
  volume     = {57},
  doi        = {10.1109/TAC.2011.2179349},
  readstatus = {read},
}

@Article{lecleach2022algames,
  author    = {Le Cleac'h, Simon and Schwager, Mac and Manchester, Zachary},
  journal   = {Autonomous Robots},
  title     = {ALGAMES: a fast augmented Lagrangian solver for constrained dynamic games},
  year      = {2022},
  issn      = {1573-7527},
  number    = {1},
  pages     = {201--215},
  volume    = {46},
  doi       = {10.1007/s10514-021-10024-7},
  publisher = {Springer},
  refid     = {Le Cleac'h2022},
  url       = {https://doi.org/10.1007/s10514-021-10024-7},
}

@InBook{carlson1995turnpike,
  author    = {Carlson, D. and Haurie, A.},
  pages     = {353--376},
  publisher = {Birkh{\"a}user Boston},
  title     = {A Turnpike Theory for Infinite Horizon Open-Loop Differential Games with Decoupled Controls},
  year      = {1995},
  isbn      = {9781461242741},
  booktitle = {New Trends in Dynamic Games and Applications},
  doi       = {10.1007/978-1-4612-4274-1_18},
}

@Article{carlson1996turnpike,
  author    = {Carlson, D. and Haurie, A.},
  journal   = {SIAM Journal on Control and Optimization},
  title     = {A Turnpike Theory for Infinite-Horizon Open-Loop Competitive Processes},
  year      = {1996},
  issn      = {1095-7138},
  month     = jul,
  number    = {4},
  pages     = {1405--1419},
  volume    = {34},
  doi       = {10.1137/s0363012994265158},
  publisher = {Society for Industrial \& Applied Mathematics (SIAM)},
}

@InBook{carlson2000infinite,
  author    = {Carlson, Dean A. and Haurie, Alain B.},
  pages     = {195--212},
  publisher = {Birkh{\"a}user Boston},
  title     = {Infinite Horizon Dynamic Games with Coupled State Constraints},
  year      = {2000},
  isbn      = {9781461213369},
  booktitle = {Advances in Dynamic Games and Applications},
  doi       = {10.1007/978-1-4612-1336-9_10},
}

@Article{hall2024receding,
  author    = {Hall, Sophie and Guerrini, Laura and D{\"o}rfler, Florian and Liao-McPherson, Dominic},
  journal   = {IFAC-PapersOnLine},
  title     = {Receding Horizon Games for Modeling Competitive Supply Chains},
  year      = {2024},
  issn      = {2405-8963},
  number    = {18},
  pages     = {8--14},
  volume    = {58},
  doi       = {10.1016/j.ifacol.2024.09.002},
  publisher = {Elsevier BV},
}

@Article{kulkarni2012variational,
  author     = {Ankur A. Kulkarni and Uday V. Shanbhag},
  journal    = {Automatica},
  title      = {On the variational equilibrium as a refinement of the generalized {N}ash equilibrium},
  year       = {2012},
  issn       = {0005-1098},
  number     = {1},
  pages      = {45-55},
  volume     = {48},
  doi        = {https://doi.org/10.1016/j.automatica.2011.09.042},
  readstatus = {skimmed},
  relevance  = {relevant}
}

@InProceedings{hall2025limits,
  author    = {Hall, Sophie and D{\"o}rfler, Florian and Nax, Heinrich H. and Bolognani, Saverio},
  booktitle = {2025 IEEE 64th Conference on Decision and Control (CDC)},
  title     = {The Limits of ``Fairness" of the Variational Generalized Nash Equilibrium},
  year      = {2025},
  pages     = {5354-5360},
  doi       = {10.1109/CDC57313.2025.11312061},
}

@Article{benenati2025linear,
  author   = {Benenati, Emilio and Grammatico, Sergio},
  journal  = {IEEE Transactions on Automatic Control},
  title    = {Linear-Quadratic Dynamic Games as Receding-Horizon Variational Inequalities},
  year     = {2025},
  pages    = {1-16},
  doi      = {10.1109/TAC.2025.3632150},
}

@Article{atzeni2013demand,
  author     = {Atzeni, Italo and Ord{\'o}{\~n}ez, Luis G. and Scutari, Gesualdo and Palomar, Daniel P. and Fonollosa, Javier Rodr{\'i}guez},
  journal    = {IEEE Transactions on Smart Grid},
  title      = {Demand-side management via distributed energy generation and storage optimization},
  year       = {2013},
  month      = jun,
  number     = {2},
  pages      = {866-876},
  volume     = {4},
  doi        = {10.1109/TSG.2012.2206060},
  publisher  = {Institute of Electrical and Electronics Engineers ({IEEE})},
  readstatus = {read},
}

@Article{li2025turnpike,
  author        = {Li, Xun and Wu, Fan and Zhang, Xin},
  title         = {Turnpike properties for zero-sum stochastic linear quadratic differential games of Markovian regime switching system},
  year          = {2025},
  month         = sep,
  archiveprefix = {arXiv},
  copyright     = {arXiv.org perpetual, non-exclusive license},
  doi           = {10.48550/ARXIV.2509.09358},
  eprint        = {2509.09358},
  primaryclass  = {math.OC},
  publisher     = {arXiv},
}

@Article{cohen2025turnpike,
  author        = {Cohen, Asaf and Jian, Jiamin},
  title         = {Turnpike properties in linear quadratic Gaussian N-player differential games},
  year          = {2025},
  month         = jul,
  archiveprefix = {arXiv},
  copyright     = {Creative Commons Attribution 4.0 International},
  doi           = {10.48550/ARXIV.2507.11632},
  eprint        = {2507.11632},
  primaryclass  = {math.OC},
  publisher     = {arXiv},
}

@Article{ersland2025long,
  author        = {Ersland, Olav and Jakobsen, Espen Robstad and Porretta, Alessio},
  title         = {Long time behaviour of Mean Field Games with fractional diffusion},
  year          = {2025},
  month         = may,
  archiveprefix = {arXiv},
  copyright     = {Creative Commons Attribution 4.0 International},
  doi           = {10.48550/ARXIV.2505.06183},
  eprint        = {2505.06183},
  primaryclass  = {math.AP},
  publisher     = {arXiv},
}

@Article{fedorov2025studying,
  author    = {Fedorov, F. A.},
  journal   = {Moscow University Computational Mathematics and Cybernetics},
  title     = {Studying the Well-Posedness of the Boundary Value Problem for a System of Riccati Type Equations Based on the Concept of Mean Field Games},
  year      = {2025},
  issn      = {1934-8428},
  month     = jun,
  number    = {2},
  pages     = {150--164},
  volume    = {49},
  doi       = {10.3103/s0278641925700086},
  publisher = {Allerton Press},
}

@Article{cirant2021long,
  author    = {Cirant, Marco and Porretta, Alessio},
  journal   = {ESAIM: Control, Optimisation and Calculus of Variations},
  title     = {Long time behavior and turnpike solutions in mildly non-monotone mean field games},
  year      = {2021},
  issn      = {1262-3377},
  pages     = {86},
  volume    = {27},
  doi       = {10.1051/cocv/2021077},
  editor    = {Buttazzo, G. and Casas, E. and de Teresa, L. and Glowinski, R. and Leugering, G. and Tr{\'e}lat, E. and Zhang, X.},
  publisher = {EDP Sciences},
}

@Article{carmona2024leveraging,
  author    = {Carmona, Ren{\'e} A. and Zeng, Claire},
  journal   = {IEEE Open Journal of Control Systems},
  title     = {Leveraging the Turnpike Effect for Mean Field Games Numerics},
  year      = {2024},
  issn      = {2694-085X},
  pages     = {389--404},
  volume    = {3},
  doi       = {10.1109/ojcsys.2024.3419642},
  publisher = {Institute of Electrical and Electronics Engineers (IEEE)},
}

@Article{fershtman1986turnpike,
  author    = {Fershtman, Chaim and Muller, Eitan},
  journal   = {Journal of Economic Theory},
  title     = {Turnpike properties of capital accumulation games},
  year      = {1986},
  issn      = {0022-0531},
  month     = feb,
  number    = {1},
  pages     = {167--177},
  volume    = {38},
  doi       = {10.1016/0022-0531(86)90094-3},
  publisher = {Elsevier BV},
}

@Article{Mckenzie76,
  Title                    = {Turnpike theory},
  Author                   = {McKenzie, L.W.},
  Journal                  = {Econometrica: Journal of the Econometric Society},
  Year                     = {1976},
  Number                   = {5},
  Pages                    = {841--865},
  Volume                   = {44},

  Publisher                = {JSTOR}
}

@Article{liaomcpherson2019regularized,
  author    = {Dominic Liao-McPherson and Mike Huang and Ilya Kolmanovsky},
  journal   = {IEEE Transactions on Automatic Control},
  title     = {A Regularized and Smoothed Fischer{\textendash}Burmeister Method for Quadratic Programming With Applications to Model Predictive Control},
  year      = {2019},
  issn      = {2334-3303},
  month     = jul,
  number    = {7},
  pages     = {2937--2944},
  volume    = {64},
  doi       = {10.1109/tac.2018.2872201},
  publisher = {Institute of Electrical and Electronics Engineers ({IEEE})},
}

@Article{epfl:faulwasser15h,
  author  = {Faulwasser, T. and Korda, M. and Jones, C.N. and Bonvin, D.},
  journal = {Automatica},
  title   = {On Turnpike and Dissipativity Properties of Continuous-Time Optimal Control Problems},
  year    = {2017},
  pages   = {297-304},
  volume  = {81},
  doi     = {10.1016/j.automatica.2017.03.012},
}

@Article{Kellett14,
  Title                    = {A compendium of comparison function results},
  Author                   = {Kellett, C.M.},
  Journal                  = {Mathematics of Control, Signals, and Systems},
  Year                     = {2014},
  Number                   = {3},
  Pages                    = {339--374},
  Volume                   = {26},

  Publisher                = {Springer}
}

@Article{Willems72a,
  author    = {Willems, J.C.},
  title     = {Dissipative dynamical systems part I: General theory},
  number    = {5},
  pages     = {321--351},
  volume    = {45},
  journal   = {Archive for Rational Mechanics and Analysis},
  publisher = {Springer},
  year      = {1972},
}

@Article{hall2025system,
  author        = {Hall, Sophie and D{\"o}rfler, Florian and Faulwasser, Timm},
  title         = {System-Theoretic Analysis of Dynamic Generalized Nash Equilibrium Problems -- Turnpikes and Dissipativity},
  year          = {2025},
  month         = oct,
  archiveprefix = {arXiv},
  copyright     = {Creative Commons Attribution Non Commercial No Derivatives 4.0 International},
  doi           = {10.48550/ARXIV.2510.21556},
  eprint        = {2510.21556},
  primaryclass  = {eess.SY},
  publisher     = {arXiv},
}

@Article{gruene2014asymptotic,
  author    = {Gr{\"u}ne, Lars and Stieler, Marleen},
  journal   = {Journal of Process Control},
  title     = {Asymptotic stability and transient optimality of economic {MPC} without terminal conditions},
  year      = {2014},
  month     = {aug},
  number    = {8},
  pages     = {1187--1196},
  volume    = {24},
  doi       = {10.1016/j.jprocont.2014.05.003},
  publisher = {Elsevier {BV}},
}

@Article{bemporad2025nashopt,
  author  = {A. Bemporad},
  journal = {arXiv preprint 2512.23636},
  title   = {{NashOpt}: A {Python} Library for Computing Generalized {Nash} Equilibria and Game Design},
  year    = {2025},
  note    = {\url{https://github.com/bemporad/nashopt}},
}

@InProceedings{faulwasser15approx,
  Title                    = {On the Design of Economic {NMPC} based on Approximate Turnpike Properties},
  Author                   = {Faulwasser, T. and Bonvin, D.},
  Booktitle                = {Proc. of 54th IEEE Conference on Decision and Control},
  Year                     = {2015},

  Address                  = {Osaka, Japan},
  Month                    = {December 15-18},
  Pages                    = {4964 - 4970},

  Doi                      = {10.1109/CDC.2015.7402995}
}

\end{document}